\documentclass[a4paper,10pt]{article}
\usepackage{amsmath, amsfonts, amssymb, amsthm}
\usepackage[english]{babel}
\usepackage[applemac]{inputenc}
\usepackage{palatino}
\newtheorem{Definition}{Definition}
\newtheorem{Proposition}{Proposition}

\newtheorem{Lemma}{Lemma}

\usepackage{graphicx}

\theoremstyle{definition}

\begin{document}

\begin{titlepage}
\title{Haantjes manifolds with symmetry} 
\author{F. Magri\\
\small{Dipartimento di Matematica ed Applicazioni, Universita' di Milano Bicocca,}\\ \small{20125 Milano, Italy}}
\date{December 06, 2017}
\maketitle

\begin{abstract}
This paper has two purposes. The first is to introduce the definition of Haantjes manifolds with symmetry. The second is to explain why these manifolds appear in the theory of integrable systems of hydrodynamic type and in topological field theories.
\end{abstract}
\end{titlepage}

\section{Introduction}
In 1968 K. Yano and M. Ako significantly extended the work of Schouten and Nijenhuis on differential concomitants. Several years earlier, in 1940 and in 1951, Schouten and Nijenhuis had discovered two remarkable differential concomitants, nowadays called the Schouten bracket and the Nijenhuis torsion, associated with a skewsymmetric tensor field of type $(2,0)$ and with a tensor field of type $(1,1)$ respectively. In  \cite{YA} Yano and Ako found the analogs of these concomitants for a wide class of higher-order tensor fields. In particular,  they noticed that if  $ C_{jk}^{l}(x) $ are the components of a tensor field $ C$  of type $(1,2)$  on a manifold $M$,  the functions

\begin{equation}\label{YanoAko}
\begin{aligned}
 \left[ C,C \right] _{jklr}^{m} :=&\sum_{s=1}^{n} \left(
C_{sj}^{m}\frac{\partial C_{lr}^{s}}{\partial {x}_{k}}+
C_{sk}^{m}\frac{\partial C_{lr}^{s}}{\partial {x}_{j}}-
C_{sr}^{m}\frac{\partial C_{jk}^{s}}{\partial {x}_{l}}-   
C_{sl}^{m}\frac{\partial C_{jk}^{s}}{\partial {x}_{r}}\right.\\
&\left.+\frac{\partial C_{jk}^{m}}{\partial {x}_{s}} C_{lr}^{s}-
\frac{\partial C_{lr}^{m}}{\partial {x}_{s}} C_{jk}^{s}  \right) .
\end{aligned}
\end{equation}
are the components of a tensor field $\left[ C,C \right] $  of type $(1,4)$, provided that the components of $C$ satisfy the symmetry conditions

\begin{equation}\label{commutativity}
 C_{jk}^{l}  =  C_{kj}^{l}  
\end{equation}
and the associativity conditions
\begin{equation}\label{associativity}
\sum_{l=1}^{n} C_{jk}^{l} C_{lm}^{s} = \sum_{l=1}^{n} C_{mk}^{l} C_{lj}^{s} .
\end{equation}
The new tensor field $\left[ C,C \right] $ may be referred to as the Yano-Ako bracket of $C$ with itself. Its discovery was the result of a cumbersome computation in the style of the old tensor calculus. Yano and Ako started by considering different Lie derivatives of $C$ along several vector fields, and managed to combine them in such a way as to cancel all the terms containing derivatives of the components of the vector fields. 
They realized that the cancellation could  be brought to an end if the tensor field $C$ obeyed the algebraic constraints written above. The new object lacked any geometric interpretation, and consequently it was difficult to foresee  possible uses for it . For this reason, the Yano-Ako bracket did not attract much attention and was rapidly forgotten.

A surprising application of the above bracket has been found many years later in the theory of integrable systems of hydrodynamic type .  The study of the universal Whitham equations introduced by  Krichever \cite{Krich} has led, in 2007, Boris Konopelchenko and me to consider a special class of deformations  of associative and commutative algebras, called coisotropic deformations \cite{KoMa}. We found that these deformations were controlled by a remarkable set of differential equations, which we called the ``central system''. 
At that time we were unaware of the work of Yano and Ako, but we quickly realized that the central system was nothing but the vanishing of the Yano-Ako bracket. The effect was to attract our attention to the work of these  authors, and to convince ourselves of the importance of their bracket in the theory of integrable systems of hydrodynamic type . Accordingly, we began to look for a geometric interpretation of their bracket. The outcome of the ensuing work is the 
concept of Haantjes manifold discussed in the present paper.

Haantjes manifolds are a tool in the analysis of the foundations of the theory of integrable systems of hydrodynamic type from a geometric point of view. The concept of Haantjes manifold may help to understand what is the minimal system of assumptions  to be set at the basis of the theory, and what is the role of each separate assumption. It can be regarded as a ramification of the concept of bihamiltonian manifold. 
The main novelty is the construction leading to a "square of exact 1-forms" . This construction extends the recursion procedures of bihamiltonian geometry. The square of 1-forms recovers many interesting integrability conditions which had already appeared in different contexts. Among them: the Yano-Ako equations of the theory of deformations; the WDVV equations of topological field theories \cite{Witten}, \cite{DVV}; the integrability conditions for the multiplicative structure of a Frobenius manifold  \cite{Dub}, \cite{Dub2}; the integrability conditions of the theory of semihamiltonian systems 
of hydrodynamic type \cite{Tsa}. 
The Haantjes manifolds endowed with symmetries have also an interesting link with Riemannian geometry.

This paper consists of three sections. Sec.~2 presents the concept of Haantjes manifold. Sec.~3 shows the links of these manifolds  with the Yano-Ako equations, with the WDVV equations, and with the theory of integrable systems of hydrodynamic type. Sec.~4, finally, establishes the link with Riemannian geometry.

\section{Haantjes manifolds}

Let us consider a manifold $M$, of dimension $n$, equipped with an exact 1-form $dA$ and with a tensor field $K : TM\rightarrow TM$ of type $(1,1)$. 
It is convenient to regard $K$ as a vector-valued 1-form on $M$, and to denote by $d_{K}$ the derivation on forms associated to $K$ according to the theory of derivations of Fr\"olicher and Nijenhuis \cite{FrNi}.  The tensor field $K$ naturally acts  on the 1-form $dA$ , mapping it into a new 1-form denoted by $KdA$. It often happens, for a wide choice of $dA$ and $K$, that the new 1-form $KdA$ is still exact . The exactness condition is the weak cohomological condition 
\begin{equation}
dd_{K} A =0, 
\end{equation}
which takes the form of a Euler-Poisson-Darboux system of partial differential equations in a suitable system of coordinates. By repeating the process, one finds that the 1-form $K^{2}dA$ is seldom exact. There is a new strong obstruction, represented by the vanishing of the   2-form
\begin{equation}
d_{K} d_{K} A =0 .
\end{equation}
Let us study this obstruction more attentively.

\begin{Lemma} Suppose that the first iterated 1-form $KdA$ is exact. Then the second iterated 1-form $K^{2}dA$ is exact if and only if   $d_{K} d_{K} A =0 $  or, equivalently, if and only if the 1-form $dA$ annihilates the Nijenhuis torsion of $K$, viewed as a vector-valued 2-form on $M$.
\end{Lemma}
\begin{proof} Let $\alpha$ be any 1- form on $M$, and let $\alpha'$ be its first iterated 1-form  : $\alpha' = K \alpha$. Then the following identities relate the $d$ and $d_{K}$ differentials of these 1-forms :
\begin{align*}
d\alpha'(\xi,\eta)&=d\alpha(K\xi,\eta)+d\alpha(\xi,K\eta)-d_{K}\alpha(\xi,\eta)\\
d_K \alpha'(\xi,\eta)&=d\alpha(K\xi,K\eta)+\alpha(T_K(\xi,K\eta))  .
\end{align*}
Here the symbol $T_{K}(\xi,\eta)$ denotes the Nijenhuis torsion of $K$ , viewed as a vector-valued 2- form  evaluated on the arbitrary pair of vector fields $\xi$ and $\eta$. The above identities can be used as follows. Take  $\alpha=K dA=d_{K} A$ in the first identity. It immediately shows that $\alpha' = K^{2}dA$ is exact if and only if  $d_{K} d_{K} A =0 $. Take $\alpha=dA$ in the second identity. It immediately shows that  $d_{K} d_{K} A =0 $  if and only if $dA$ annihilates the torsion of $K$,  as claimed.
\end{proof}

This Lemma explains why the Nijenhuis torsion of $K$ plays a prominent role in the theory of recursion operators. It points out that there are only two  possibilities. If the Nijenhuis torsion of $K$ vanishes, there are no  other obstructions  to the process of iteration. All the iterated 1-forms $KdA$, $K^{2}dA$, $K^{3}dA$, and so on are exact. 
If the Nijenhuis torsion of $K$ does not vanish, the process of iteration ends after two steps, because the 1-form $K^{3}dA$ cannot be exact. A  way of circumventing this obstruction is to renounce  the idea of a single recursion operator, and to consider a more general  scheme where several recursion operators act at the same time. They should not be the powers of a single recursion operator.  To see how to manage the new situation, and to find the right conditions, let us look again at the case of a single recursion operator from a different standpoint. Let us agree to denote the first $n$ powers of $K$ by the 
symbols
\begin{equation*}
K_{1} = Id, \quad  K_{2} =K, \quad K_{3} = K^{2}, \cdots, K_{n} = K^{n-1},
\end{equation*}
to prepare the transition to the general case. It is clear from what has been said before that the doubly  iterated 1-forms $K_{j} K_{l} dA$ are exact,  since $ K_{j} K_{l}= K^{j+l-2} $. Of course also the triply iterated 1- forms $K_{j} K_{l} K_{m} dA$ are exact , but  it is wise to ignore this fact. Indeed, to insist on it would lead one to fall again into the old case. So, the right idea   is to work with  a family of $n$ distinct recursion operators $(K_{1} = Id,\ K_{2}, \ K_{3},  \cdots, K_{n} )$ which behave like the first $n$ powers of $K$  up to the second iteration, without being the powers of $K$. This idea is formalized in the following definition of Haantjes manifold.

\begin{Definition}  Consider a manifold $M$, of dimension $n$, equipped with an  exact 1-form $dA$. Assume that  $(K_{1} = Id,\ K_{2}, \ K_{3},  \cdots, K_{n} )$  are $n$ pairwise commuting tensor fields of type $(1,1)$  on $M$:
\begin{equation}
K_{j} K_{l} = K_l K_{j}  .
\end{equation}
The manifold $M$ is a Haantjes manifold if all the doubly iterated 1-forms  $K_{j} K_{l} dA$ are closed and therefore locally exact. These forms constitute the square of 1-forms  of the Haantjes manifold.
Since we shall limit ourselves to the study of the local geometry of the manifold, we shall always admit that the 1-forms $K_{j} K_{l} dA$  are exact. We set
\begin{equation}
K_{j} K_{l} dA = d A_{j l}  ,
\end{equation}
and we call the scalar functions $A_{j l} $ the potential functions of the Haantjes manifold. They form a symmetric matrix $H$ referred to as the matrix potential of the manifold.

\end{Definition}
It should be clear from the previous Lemma  that the recursion operators $K_{j}$  of a Haantjes manifold cannot be chosen arbitrarily. They must verify suitable integrability conditions, weaker than the vanishing of the Nijenhuis torsion, coming from the exactness condition for the 1-forms $dA_{j l} $. The discussion of the full set of integrability conditions is a delicate problem, which goes beyond the scope of the present paper. However, one basic condition must be mentioned.

\begin{Proposition} Assume that at least one of the recursion operators $K_{j}$ has real and distinct eigenvalues. Then the Haantjes torsion of all the recursion operators $K_{j}$ vanishes.
\end{Proposition}

Let us recall that the Haantjes torsion of a tensor field $K$, of  type $(1,1)$, is a vector-valued 2-form, related to the Nijenhuis torsion of $K$ according to
\begin{equation*}
H_K(\xi,\eta)=T_K(K\xi,K\eta)-KT_K(K\xi,\eta)-KT_K(\xi,K\eta)+K^2T_K(\xi,\eta).
\end{equation*}
The vanishing of the Haantjes torsion is the necessary and sufficient condition for the integrability of the eigendistributions of the tensor field $K$,  as shown by Haantjes in 1955 \cite{Haa} ( under the semisimplicity assumption stated above). 
Therefore, if at least one of the recursion operators of a Haantjes manifold has real and distinct eigenvalues, on the manifold there exists a privileged system of coordinates in which all the recursion operators become diagonal. These coordinates are usually called canonical coordinates (or Riemann invariants) in the theory of the integrable systems of hydrodynamic type.

The above Proposition will not be proved in this paper, since the proof is long, and since the result will not be used afterwards. Its only use is to justify the name of Haantjes manifolds given the manifolds defined above , and to motivate the introduction of the notion of weak Haantjes manifold.

\begin{Definition} A weak Haantjes manifold is a manifold $M$ equipped with a single  exact  1-form $dA$ and with a single tensor field $K$ of type $(1,1)$  satisfying the following three conditions:
\begin{align}
Haantjes( K) &= 0  \\
dd_{K}A &=0\\
d_{K}d_{K}A &= 0
\end{align}

\end{Definition}

The weak Haantjes manifold seems to be the \emph{minimal} and hence \emph{basic geometric structure} underlying the theory of recursion operators. In many examples one finds that  it is possible to extend a weak Haantjes manifold into a full Haantjes manifold by recovering the missing $(n-2)$  recursion operators directly from $K$. This happens, for instance, when the Nijenhuis torsion of $K$ has particularly nice forms. In these cases it is possible to construct the missing tensor fields  $K_{j}$ as polynomial functions of $K$.  
One is, thus,  almost back to the initial situation, when the Nijenhuis torsion of $K$ was supposed to vanish.  The main difference is that one has lost the rule of powers, to which are substituted suitable polynomials in $K$ constructed, case by case,  in such a way as to compensate for the non-vanishing of the torsion of $K$. The study of the problem of the extension of a weak Haantjes manifold into a full Haantjes manifold is a fascinating problem that leads to numerous 
interesting results. However it requires going deeper into the problem of 
classification of  the weak Haantjes manifolds. For the moment there is no more than a reasoned collection of  examples. I have hesitated to give a definition of weak Haantjes manifold. 
However, I may remark that if one adds the condition that the tensor field $K$ can be diagonalized, one readily finds that the conditions $dd_{K}A = 0 $ and $d_{K}d_{K}A = 0$ permit recovering the definition of semihamiltonian systems given by Tsarev in canonical coordinates. Thus, the definition of weak Haantjes manifold can be seen as an intrinsic formulation of  Tsarev's theory. I am convinced that the concept of weak Haantjes manifold has  a central position in the present theory.

\section{Three properties of Haantjes manifolds}

The purpose of this section is to outline the links among the Haantjes manifolds, the Yano-Ako differential concomitant, the topological field theories, and the integrable systems of hydrodynamic type.

The relation with Yano and Ako is quite simple. It is based on the  remark that the recursion operators $K_{j}$ of a Haantjes manifold form an associative and commutative algebra with unity. To prove this property, expand  the 1-form  $K_{j} K_{l} dA$ on the basis of 1-forms $dA_{m} = K_{m} dA$.  Call $ C_{jl}^{m}(A_{i}) $ its components:
\begin{equation}
K_jK_ldA=\sum C^{m}_{jl}dA_m
\end{equation}
By acting with the recursion operator $K_{n}$  on both sides of this equation,  infer the identity
\begin{equation*}
K_{j} K_{l} dA_{n}=\sum C^{m}_{jl}K_{m} dA_{n}.
\end{equation*}
It allows us to conclude that
\begin{equation}
K_{j} K_{l} =\sum C^{m}_{jl}K_{m},
\end{equation}
since the 1-forms $dA_{n} $ form a basis. This relation proves that the recursion operators of a Haantjes manifold form an associative algebra. To proceed towards the Yano-Ako equations, let us notice that the structure constants $ C_{jl}^{m}(A_{i}) $  of this algebra are the partial derivatives of the potential functions  $A_{j l}$ with respect to the coordinates $A_{m}$, as shown by their  definition. Write the Yano-Ako equations in these coordinates, and notice that all the terms cancel in pairs owing to the above property. 
Thus the Yano-Ako equations hold true in the coordinates $A_{m}$. Since they are tensorial , they hold true in any coordinate system. The conclusion is that the structure constants of the algebra of recursion operators  of a Haantjes manifold satisfy  the Yano-Ako equations (or  central system, in the terminology of \cite{KoMa}). This result provides a class of solutions of the Yano-Ako equations having a geometric meaning, but it does not yet solve completely the problem stated in the 
Introduction. It remains unclear how exhaustive this class of solutions may be.

The relation with topological field theories involves the potential functions $A_{j l}$. They are scalar functions on the manifold, and accordingly they can be written in any coordinate system. Nevertheless, the recursion operators select a class of special coordinates on the Haantjes manifold making manifest a rather special property of these functions. To work out this property we need the concept of generator of a Lenard chain.

\begin{Definition} A generator of a Lenard chain on a Haantjes manifold is a vector field $\xi$ such that the iterated vector fields $\xi_{j} = K_{j} \xi$ are linearly independent and commute in pairs.
\end{Definition}

Assume, for the moment, the existence of such a generator, and notice that it provides a distinguished system of coordinates on the Haantjes manifold, since the vector fields $\xi_{j}$ commute. Call  $t_{j}$ the corresponding coordinates:
\begin{equation*}
\xi_j=\frac{\partial}{\partial t^j}.
\end{equation*}
Write the potential functions in these coordinates. Then the following remarkable property holds true.

\begin{Proposition}\label{WDVV}
In the coordinates defined by the generator of a Lenard chain, the matrix of the potential functions of a Haantjes manifold is the Hessian matrix of a function $F(t_{1}, t_{2}, \cdots , t_{n} ) $. This function is a solution of the (generalized) WDVV equations of topological field theories. Any solution of the WDVV equations can be obtained in this way.
\end{Proposition}

This proposition  has been proved in \cite{Ma} . It subordinates the existence of the function $F(t_{1}, t_{2}, \cdots , t_{n} ) $ to the existence of a generator of a Lenard chain. This problem leads us to the theory of the integrable sytems of hydrodynamic type.

It is well known that there is a one-to-one correspondence between systems of  equations of hydrodynamic type and tensor fields of type $(1,1)$, such as $K$. To pass from the tensor field to the differential equations, it is enough to introduce any coordinate system $u^{j}$ on the manifold,  and to consider the corresponding components of the tensor  field $K$, defined by
\begin{equation}
Kdu^j=\sum K^j_l(u) du^l .
\end{equation}
Then the equations of hydrodynamic type are written in the form
\begin{equation}
\frac{\partial u^j}{\partial t}=\sum K^j_l(u) \frac{\partial u^l}{\partial x}.
\end{equation}
By inverting the steps, one easily passes from the differential equations to the tensor field $K$. The tensorial character of $K$ is guaranteed by the transformation law of the system of differential equations under a change of the unknown functions. 
It is fair to say that the tensor field $K$ gives an intrinsic description of the differential equations, which allows to control the properties of the equations in any coordinate system. On a Haantjes manifold one has $n$ tensor fields $K_{j}$, and therefore $n$ systems of differential equations of hydrodynamic type, each composed of $n$ differential equations.

\begin{Proposition}
The $n$ systems of differential equations of hydrodynamic type associated with the tensor fields $K_{1}=Id, K_{2} , \cdots , K_{n} $ of a Haantjes manifold are mutually compatible, and therefore there exists a solution $ u^{j}( t_{1}, t_{2}, \cdots , t_{n}) $ common to all of them. Furthermore, in the system of coordinates $A_{m}$  the differential equations take the form of  conservation laws.
\end{Proposition}

\begin{proof} To prove the compatibility of the $n$ systems of PDEs, it is necessary and sufficient to prove that the tensor fields $K_{j}$ satisfy the following identity 
\begin{equation}
[K_{j}\xi, K_{l}\xi]-K_j[\xi, K_{l}\xi]-K_{l}[K_{j}\xi, \xi]=0
\end{equation}
for any choice of the vector field $\xi$. This identity assures the equality of the second-order mixed derivatives of the field functions $u^{j}$  with respect to the independent variables  ${t^{k}}$, on account  of the commutativity of the tensor fields $K_{j}$. 
To prove the identity, it is useful to evaluate the above vector expression on the basis of the differentials $dA_{m}$, in order to use the basic relation $K_{j} dA_{m} = dA_{j m} $. Let us denote by $\xi_{j}$, as before,  the vector field $K_{j} \xi$, knowing that these vector fields do not commute  since $\xi$ is  not assumed to be a generator of a Lenard chain. Keep in mind that $\xi_{l} (A_{j m} ) - \xi_{j} (A_{l m} ) = 0$ since the tensor fields $K_{j}$ commute.  The identity is then proved as follows:
\begin{align*}
&dA_m([K_{j}\xi, K_{l}\xi]-K_j[\xi, K_{l}\xi]-K_{l}[K_{j}\xi, \xi])\\
&=\xi_j\xi_l(A_m)-\xi_l\xi_j(A_m)-\xi \xi_l(A_{jm}) + \xi_l \xi (A_{jm})-\xi_j\xi(A_{lm})+\xi\xi_j(A_{lm})\\
&=\xi_j\xi(A_{lm})-\xi_l\xi(A_{jm})-\xi \xi_l(A_{jm}) + \xi_l \xi (A_{jm})-\xi_j\xi(A_{lm})+\xi\xi_j(A_{lm})\\
&=\xi(\xi_j (A_{lm})-\xi_j(A_{jm}))=0.
\end{align*}
The existence of a common solution is thus established. To see that the differential equations can be written as conservation laws, it is enough to write them explicitly in the $A_{m}$ coordinates and to use again the basic relation  $K_{j} dA_{m} = dA_{j m} $. 
\end{proof}

Now we are in a position to discuss the problem of the existence of the generators of Lenard chains on a Haantjes manifold. The tool is the common solution $ u^{j}( t_{1}, t_{2}, \cdots , t_{n}) $ of the differential equations, whose existence has just been established. Let us regard this solution as the definition of a change of coordinates on the manifold $M$, from the old coordinates $ u^{j} $ to the new coordinates $t^{k}$.  
Let us denote by $\dfrac{\partial}{\partial t^{k}}$ the vector fields of the corresponding  basis in $TM$. It is almost a tautology to see that these vector fields form a Lenard chain, due to the form of the differential equations. Thus, one can say that  there is a one-to-one correspondence between the solutions of the systems of differential equations of hydrodynamic type associated with the tensor fields $K_{j}$ and the Lenard chains of vector fields on a Haantjes manifold. Combined with   Proposition \ref{WDVV} , this remark shows that a solution of the WDVV equations is 
associated 
with any solution of the system of hydrodynamic type (and viceversa). This is one of the possible ways of introducing the Hirota tau function in the present framework.

\section{Haantjes manifolds with symmetry}

There is a second class of vector fields worth attention on a Haantjes manifold, apart from the generators of Lenard chains. They are the conformal symmetries of the manifold. 

\begin{Definition}
A vector field $\xi$  such that the Lie derivatives of the 1-form $dA$ and of the tensor fields $K_{j}$ along $\xi$ are  multiples of $dA$ and $K_{j}$ respectively,
\begin{gather}
Lie_{\xi}(dA)=\alpha\cdot dA\\
Lie_{\xi}(K_j)=\gamma_j \cdot K_j  ,
\end{gather}
 is called a conformal symmetry of the Haantjes manifold. It is a symmetry if the functions $\alpha$ and  $\gamma_{j}$   vanish.
\end{Definition}

As before, we denote by $\xi_{j}$ the vector fields $K_{j} \xi$ . They form a basis in $TM$, without defining a system of coordinates on $M$  since they do not commute. We use this basis, the conformal symmetry, and the potential functions $A_{j l} $ to define a second-order symmetric tensor field on $M$ by setting:

\begin{equation}
g ( \xi_{j} , \xi_{l} ) = \xi( A_{j l } ) .
\end{equation}
Explicitly, this means that the components of the tensor field $g$ on the basis $\xi_{j}$ are the derivatives of the potential functions $A_{j l} $ along the conformal symmetry $\xi$. In this section we prove the following remarkable property of this tensor field.

\begin{Proposition}\label{pr:RiemannMetric}
 Assume that the matrix $\xi( A_{j l } )$  is nonsingular, and that the functions $\alpha$ and $\gamma_j$ are constant. Then $g$ is a flat semiriemannian metric on $M$.
 \end{Proposition}

The proof of this Proposition is split into four lemmas. The statement of these lemmas is made easier by introducing the symbols $g_{jl}$ for the components $g( \xi_{j},\xi_{l})$ of the metric, and the functions 
\begin{equation}
c_{jlm}:=\xi_j(A_{lm})=\xi_l(A_{mj})=\xi_m(A_{jl}),
\end{equation}
as shorthand notation for the derivatives of the potential functions $A_{jl}$ along the vector fields $\xi_{m}$ of the basis generated by the conformal symmetry.

The first lemma yields an expression for the commutators of the vectors $\xi_j$ of the basis.
\begin{Lemma}\label{lemma:commutatore}
 \begin{equation}
  [\xi_j,\xi_l]=(\gamma_l-\gamma_j) \sum c_{jlm}\frac{\partial}{\partial A_m}.
 \end{equation}
\end{Lemma}

\begin{proof}
Since $\xi$ is a conformal symmetry, 
\begin{equation*}
[\xi,\xi_j]=Lie_{\xi}(K_j\xi)=\gamma_j\xi_j.
\end{equation*}
Consequently:
\begin{align*}
 dA_m([\xi_j,\xi_l])&= \xi_j\xi_l(A_m)-\xi_l\xi_j(A_m)\\
 &=\xi_j \xi(A_{lm}) -\xi_l \xi(A_{jm}) \\
 &=[\xi_j,\xi] (A_{lm}) -[\xi_l, \xi](A_{jm})\\
 &=(\gamma_l-\gamma_j) c_{jlm}
\end{align*}

\end{proof}

The second lemma specifies the value of the derivatives of the components of the metric along the vectors fields $\xi_j$, and also the value of the metric on the commutators $[\xi_j,\xi_l]$.

\begin{Lemma}\label{lemma:derivata}
 \begin{gather}
  g(\xi_m,[\xi_j,\xi_l]) = (\gamma_l-\gamma_j)c_{jlm}\\
  \xi_m (g_{jl}) = (\alpha +\gamma_j+\gamma_j)c_{jlm}
 \end{gather}
\end{Lemma}
\begin{proof}
The first equation is a simple consequence of the first Lemma and of the formula :
 \begin{equation}\label{eq:basis}
  \xi_j = \sum g_{jm}\frac{\partial}{\partial A_m} ,
 \end{equation}
giving the expansion of the vector fields $\xi_j$ on the basis associated to the coordinates $A_m$.
This expansion follows immediately from  the definition of the vector fields $\xi_j$ and of the metric $g$ .

To prove the second equation one notices that:
\begin{equation*}
 Lie_{\xi} d A_{jl} = Lie_{\xi}(K_jK_l dA) = (\alpha+\gamma_j+\gamma_l)dA_{jl}.
\end{equation*}
Consequently:
\begin{equation*}
 d g_{jl} = d \xi (A_{jl}) = Lie_{\xi} d A_{jl}  = (\alpha+\gamma_j+\gamma_l)dA_{jl}.
\end{equation*}
The last equation gives the statement, since : 
\begin{equation*}
 \xi_m(g_{jl})  = (\alpha+\gamma_j+\gamma_l)\xi_m(A_{jl}).
\end{equation*}
\end{proof}
\medskip

The first two lemmas allow us to compute the coefficients of the Levi-Civita connection of $g$ on the basis $\xi_j$. One must use the Koszul formula \cite[p. 61]{ON}:
\begin{equation*}
\begin{aligned}
 2g(\nabla_{\xi_j}\xi_l,\xi_m) &= \xi_j g(\xi_l,\xi_m) + \xi_l g(\xi_j,\xi_m)- \xi_m g(\xi_j,\xi_l) \\
 &\ - g(\xi_j,[\xi_l,\xi_m]) + g(\xi_l,[\xi_m,\xi_j]) + g(\xi_m,[\xi_j,\xi_l])
\end{aligned}
\end{equation*}
\medskip
\begin{Lemma}\label{lemma:nabla}
 The coefficients of the Levi-Civita connection on the basis $\xi_j$ are given by the formula:
 \begin{equation} 
  \nabla_{\xi_j}\xi_l=(\frac{\alpha}{2}+\gamma_l)\sum c_{jlm}\frac{\partial}{\partial A_m}
 \end{equation}
or equivalently, by the formula:
\begin{equation}
 g(\nabla_{\xi_j}\xi_l,\xi_m) = (\frac{\alpha}{2}+\gamma_l) c_{jlm}.
\end{equation}
\end{Lemma}

\begin{proof}
 The proof is a simple application of the previous two Lemmas and of the Koszul formula.
\end{proof}
\medskip
We are now in a position to compute finally the Riemann tensor
\begin{equation*}
 R_{\xi_j\xi_l}(\xi_m)=(\nabla_{\xi_j}\nabla_{\xi_l}-\nabla_{\xi_l}\nabla_{\xi_j}-\nabla_{[\xi_j,\xi_l]})(\xi_m).
\end{equation*}

\medskip
\begin{Lemma}
 The covariant components of the Riemann tensor on the basis $\xi_j$ are given by:
 \begin{equation}
  R_{mpjl} = (\frac{\alpha}{2} +\gamma_m)(\frac{\alpha}{2} +\gamma_p) \sum_{s,t} g^{st}(c_{jms}c_{lpt}-c_{jpt}c_{lms}).
 \end{equation}

\end{Lemma}
\begin{proof}
Let us  split the computation of the Riemann tensor into two parts. First, one considers the term $g(\nabla_{[\xi_j,\xi_l]}\xi_m,\xi_p)$. One obtains:
\begin{align*}
&g(\nabla_{[\xi_j,\xi_l]}\xi_m,\xi_p)\\
&= \sum_s g(\nabla_{(\gamma_l- \gamma_j)(c_{jls}\frac{\partial}{\partial A_s})}\xi_m,\xi_p)\\
&= \sum_{s,t}((\gamma_l- \gamma_j)c_{jls}g(\nabla_{g^{st}\xi_t}\xi_m,\xi_p)\\
&= \sum_{s,t}((\gamma_l- \gamma_j)c_{jls}g^{st}g(\nabla_{\xi_t}\xi_m,\xi_p)\\
&= (\gamma_l- \gamma_j)(\frac{\alpha}{2}+\gamma_m)\sum_{s,t}g^{st}c_{jls}c_{tmp}
\end{align*}
according to Lemma \ref{lemma:nabla}. Then, one considers the two remaining terms. By using again the properties of the connection 1-form, formalized by the Koszul axioms \cite[p.59]{ON}, and by exploiting for the first time the assumption that the functions $\alpha$ and $\gamma_j$ are constant, one finds
\begin{align*}
 &g((\nabla_{\xi_j}\nabla_{\xi_l}-\nabla_{\xi_l}\nabla_{\xi_j})\xi_m,\xi_p)\\
 = &(\frac{\alpha}{2}+\gamma_m)(\xi_j(c_{lmp})-\xi_l(c_{jmp})) \\
 &- (\frac{\alpha}{2}+\gamma_m)(\frac{\alpha}{2}+\gamma_p)\sum_{st}g^{st}(c_{lms}c_{jpt}-c_{jms}c_{lpt}) .
\end{align*}
One may simplify this expression by noticing that:
\begin{equation*}
\xi_j(c_{lmp})-\xi_l(c_{jmp}) = (\gamma_l-\gamma_j)\sum_{s,t}c_{jls}g^{st}c_{mpt}.
\end{equation*}
Indeed:
\begin{align*}
 &\xi_j(c_{lmp})-\xi_l(c_{jmp}) \\
 & = \xi_j\xi_l (A_{mp})-\xi_l\xi_j(A_{mp})\\
 & = [\xi_j,\xi_l](A_{mp}) \\
 & = (\gamma_l-\gamma_j)\sum_{s} c_{jls}\frac{\partial A_{mp}}{\partial A_{s}}\\
 & = (\gamma_l-\gamma_j)\sum_{s,t} c_{jls}g^{st}\xi_t(A_{mp})\\
 & = (\gamma_l-\gamma_j)\sum_{s,t} c_{jls}g^{st} c_{mpt} .
\end{align*}
By adding the two terms of the Riemann tensor with the proper sign, one finally obtains  the expression of its covariant components as desired.
\end{proof}
\medskip
The vanishing of the Riemann tensor is now a consequence of the fact that the recursion operators form a commutative and associative algebra.

\begin{proof}[Proof of Proposition \ref{pr:RiemannMetric}]
The Riemann tensor contains the following expression
$\sum_{s,t} g^{st}(c_{jms}c_{lpt}-c_{jpt}c_{lms}).$
Notice that 
\begin{equation*}
 \sum_{s}g^{st}c_{jms}=\sum_s g^{ts}\xi_s(A_{jm}) = \frac{\partial A_{jm}}{\partial A_t} = C^{t}_{jm}.
\end{equation*}
Therefore:
\begin{align*}
 &\ \sum_{s,t}g^{st}c_{jms}c_{tlp}\\
 =& \sum_{t}C^{t}_{jm}\xi_{p}(A_{lt})\\
 =& \sum_{t,q}C^{t}_{jm}g_{pq}\frac{\partial A_{lt}}{\partial A_q}\\
 =& \sum_{t,q}g_{pq}C^{q}_{lt}C^{t}_{jm}.
\end{align*}
and consequently:
\begin{equation*}
 R_{mpjl} = (\frac{\alpha}{2} +\gamma_m)(\frac{\alpha}{2} +\gamma_p) \sum_{q,t} g_{pq}(C^{q}_{lt}C^{t}_{jm}-C^{q}_{jt}C^{t}_{lm}).
\end{equation*}
This expression vanishes on account of the associativity condition \eqref{associativity} satisfied by the structure constants of the algebra of the recursion operators.
\end{proof}

Since it is well-known that flat Riemannian metrics define Poisson brackets for systems of differential equations of hydrodynamic type, this result leads us back to the bihamiltonian setting which was our point of departure. In some sense, the circle has been closed. 

\section{Concluding remarks}

This paper aimed at explaining the role of Haantjes manifolds in the theory of the integrable systems of hydrodynamic type and related fields. The main novelty presented here is the square of 1-forms $d A_{j l } $. It is a simple but non trivial  extension of the concept  of bihamiltonian recurrence, which seems to have passed unnoticed so far. 
As shown in this paper, the square of 1-forms recovers many interesting integrability conditions which had already appeared in different contexts. Among them: the Yano-Ako equations of the theory of deformations; the WDVV equations of topological field theories; the integrability conditions of the theory of semihamiltonian systems of Tsarev.
All these integrability conditions have been already thoroughly studied in the past, in particular by Boris Dubrovin in his theory of Frobenius manifolds. Repetitions are therefore unavoidable. Nevertheless, I hope that the geometric framework of the Haantjes manifolds provides a new view of old things, 
and that it allows to see better what is the minimal system of assumptions  to be set at the basis of the theory, and what is the role of each separate assumption. For instance, it shows that the role of the metric is not so essential in understanding the WDVV equations. The points of contact  and the differences with the previous theories will be discussed elsewhere.

\section{Appendix}
In this appendix I recall the definition of the operarator $d_{K}$ and I  exhibit a few of its interesting properties, in order to make the paper reasonably self-contained. 
I take also the opportunity of pointing out a very fine characterisation of semisimple recursion operators having vanishing Haantjes torsion discovered by Nijenhuis in 1955.

\bigskip

\textbf{1.} \emph{Definition of $d_{K}$}. According to the theory of  Fr\"olicher and Nijenhuis, the differential operator $d_{K}$ is the unique derivation, of degree 1, on the algebra of differential forms which 
satisfies the following four conditions:
\begin{align*}
& d_{K}A = K dA \\
& d_{K}( \alpha+\beta) = d_{K}\alpha + d_{K}\beta \\
& d_{K}(\alpha \wedge \beta ) = d_{K}\alpha \wedge \beta + (-1)^{a} \alpha \wedge d_{K}\beta \\
& d_{K} d + d d_{K} = 0.
\end{align*}
This definition is rather abstract, but it is easy to convert it into a powerfull algorithm to compute the differential $d_{K}$ in any concrete situation. First, one starts by writing 
the differential form $\alpha$ as a sum of products of 1-forms. 
Then one uses the second and third conditions to lead $d_{K}$ to act on any single 1-form appearing in $\alpha$. By linearity the problem is reduced to evaluate the differential of simple 1-forms of the type
 $\alpha= BdA$, where $A$ and $B$ are scalar function. This problem is solved by the first and last conditions. In this way one ends up to evaluate always differentials of scalar functions only.
 
Let us follow this procedure to prove the noticeable identity
\begin{equation*}
d_{K}^{2}A(\xi,\eta) = dA( T_{K}(\xi,\eta)).
\end{equation*}
which holds for any scalar function $A$. First we notice that
\begin{align*}
d_{K}^{2}A& = d_{K} (d_{K}A) \\
&=d_{K}\sum_{l} \frac{\partial A}{\partial x^{l}} d_{K}x^{l} \\
&=\sum_{l} d_{K}( \frac{\partial A}{\partial x^{l}}) \wedge d_{K}x^{l} +\sum_{l} \frac{\partial A}{\partial x^{l}} d_{K}^{2}x^{l} \\
&=\sum_{l<m}  \frac{\partial^{2} A}{\partial x^{l}\partial x^{m}} d_{K}x^{m} \wedge d_{K}x^{l} + \sum_{l} \frac{\partial A}{\partial x^{l}} d_{K}^{2}x^{l} \\
&=\sum_{l} \frac{\partial A}{\partial x^{l}} d_{K}^{2}x^{l}.
\end{align*}
Then we evaluate the differentials  $d_{K}^{2}x^{j}$ of the coordinate functions according to the above procedure:
\begin{align*}
d_{K}^{2}x^{j} &= d_{K}( Kdx^{j}) \\
&= \sum_{p} d_{K}(K_{p}^{j} dx^{p}) \\
&= \sum_{p} d_{K}(K_{p}^{j}) \wedge dx^{p} - K_{p}^{j} dd_{K}x^{p} \\
&= \sum_{l,m,p} \left(\frac{\partial K_{p}^{j}}{\partial x^{l}} K_{m}^{l} -K_{l}^{j}\frac{\partial K_{p}^{l}}{\partial x^{m}}\right) dx^{m} \wedge dx^{p}  .
\end{align*}
We conclude that
\begin{equation*}
 d_{K}^{2}x^{j} = \sum_{l<m} T_{l m}^{j} dx^{l} \wedge dx^{m} ,
\end{equation*}
where $ T_{l m}^{j}$ are the components of the torsion tensor $T_{K}$ of $K$. By inverting this formula, we can write the torsion tensor of $K$ in the form :
\begin{equation*}
T_{K} = \sum_{j}  d_{K}^{2}x^{j} \otimes \frac{\partial }{\partial x^{j} }.
\end{equation*} 
It implies that
\begin{equation*}
dA( T_{K}(\xi,\eta))= \sum_{l} \frac{\partial A}{\partial x^{l}} d_{K}^{2}x^{l}(\xi,\eta).
\end{equation*}
The comparison with $d_{K}^{2}A$ proves the identity mentioned above.

\medskip

\textbf{2.} \emph{Identities}. In the study of the recurrence of exact 1-forms pursued in Sec.2 we made use of the identities 
\begin{align*}
d\alpha'(\xi,\eta)&=d\alpha(K\xi,\eta)+d\alpha(\xi,K\eta)-d_{K}\alpha(\xi,\eta)\\
d_K \alpha'(\xi,\eta)&=d\alpha(K\xi,K\eta)+\alpha(T_K(\xi,K\eta))  ,
\end{align*}
relating the $d$ and $d_{K}$ differentials of any 1-form $\alpha$ to the differentials of its iterated 1-form $\alpha'= K\alpha$.
We now prove these identities. By linearity, it is sufficient to consider the special pair of 1-forms  $ \alpha= BdA $ and  $\alpha' = Bd_{K}A $ , where $A$ and $B$ are arbitrary functions.
For $d_{K}\alpha$ we have:
\begin{align*}
d_{K}\alpha &= d_{K}B \wedge dA + B d_{K}dA \\
&= d_{K}B \wedge dA - B d d_{K}A \\
&= d_{K}B \wedge dA + dB \wedge d_{K}A - d(B d_{K}A )  \\
&= d_{K}B \wedge dA + dB \wedge d_{K}A - d\alpha'  .
\end{align*}
Once evaluated on two arbitrary vector fields $\xi$ and $\eta$ this equation gives
\begin{equation*}
d_{K}\alpha(\xi,\eta) = d\alpha(K\xi,\eta) + d\alpha(\xi,\eta) - d\alpha'(\xi,\eta) .
\end{equation*}
This  is already the first identity. To prove the second identity, let us consider
\begin{equation*}
d_{K}\alpha' = d_{K}B \wedge d_{K}A + B d_{K}^{2}A  .
\end{equation*}
Once evaluated on the arbitrary pair of vector fields $\xi$ and $\eta$, this equation gives
\begin{equation*}
d_{K}\alpha' (\xi,\eta) = (dB \wedge dA)(K\xi,K\eta)+ B dA(T_{K}(\xi,\eta)).
\end{equation*}
Since $dB \wedge dA= d\alpha$ it can also be written in the form
\begin{equation*}
d_{K}\alpha' (\xi,\eta) = d\alpha(K\xi,K\eta) +\alpha( T_{K}(\xi,\eta)).
\end{equation*} 
This is the second identity.

\medskip

\textbf{3.} \emph{Recursion operators with vanishing Haantjes torsion}. To conclude this appendix, let us use the above formalism to write a result of Albert Nijenhuis,
concerning the recursion operators with vanishing Haantjes torsion, in a form which is particularly terse and useful. From the previous discussion, it is clear 
that the vanishing of the Haantjes torsion is an algebraic constraint on the Nijenhuis torsion which must be mirrored by the differential 2-form $d_{K}^{2}B$ of any function $B$. 
Assume that the recursion operator $K$ has real and distinct eigenvalues. Then according to Nijenhuis ( compare Eq.(3.10) in \cite{Nij}), there exist at most $(n-1)$  1- forms 
$ \alpha_{0}, \alpha_{1}, \dots, \alpha_{n-2}$ such that
\begin{equation*}
d_{K}^{2}B = \alpha_{0} \wedge dB + \alpha_{1} \wedge KdB + \dots + \alpha_{n-2} \wedge K^{n-2}dB .
\end{equation*}
The 1-forms are independent of the function $B$. They generate a differential ideal which, according to the result of Nijenhuis, contains the differential $d_{K}^{2}B$ of any scalar function $B$. 
This ideal is certainly 
an important element of the geometry of the recursion operator, and its study should provide clues for the classification of the  recursion operators having vanishing Haantjes torsion.
In this appendix I wish to give an example of such an ideal.

One of the simplest possible classes of recursion operators with vanishing Haantjes torsion is certainly the class of operators whose ideal  is generated by a single exact 1-form $\alpha_{0}= dA$.
This class is not void. For instance, the recursion operators associated with the Coxeter groups of type $A_{n}$ have this property. In this class of examples
\begin{equation*}
d_{K}^{2}B = dA \wedge dB
\end{equation*}  
for any function $B$. Therefore for $B=A$ one gets $d_{K}^{2}A = 0$. So the function $A$ characterizing the torsion of $K$ satisfies the strong cohomological condition $d_{K}d_{K}A=0$. 
There are cases where the function $A$ satisfies also the weak cohomological condition $dd_{K}A= 0$. These cases are clearly particularly remarkable. Indeed, without any additional assumption on $K$, 
one may implement a recursive procedure which allows to generate a sequence of functions ${A_{l}}$ satisfying the same cohomological conditions as $A$. 
The recurrence  formula is dictated by the constraint $d_{K}^{2}B = dA \wedge dB $ on the torsion of $K$. Each function $A_{l}$ defines in turn a new tensor field $K_{l}$. It is the unique tensor field 
which commute with $K$ and which maps the 1-form $dA$ into the 1-form $dA_{l}$. By this process the single operator $K$ generates an infinite sequence of operators $K_{l}$. It turns out that these 
tensor fields verify the 
conditions defining a Haantjes manifold. This is a concrete example of how a weak Haantjes manifold may be prolonged into
a Haantjes manifold when the Nijenhuis torsion of $K$ has a `` nice form'', as claimed in Sec.2. When I discussed this subject there informally ,  I had this class of examples in mind. 
I hope that the above short remarks may help to clarify the sense of that informal discussion.

\bigskip
\noindent{\bf Acknowledgements}. I wish to thank Boris Konopelchenko. Together we began the study of Haantjes manifolds.Together, I hope, we shall end it.


\begin{thebibliography}{60}


\bibitem{YA} Yano K. and Ako M.,
              \emph{On certain operators associated with tensor fields},
               Kodai Math. Sem. Rep. 20 (1968), 414-436.

\bibitem{Krich} Krichever I.M.,
              \emph{The $\tau$-function  of the universal Whitham hierarcy, matrix models and topological   fields theories},       
               Commun. Pure App. Math. 47 (1994), 437-475.


\bibitem{KoMa} Konopelchenko B. and Magri F.,
              \emph{Coisotropic deformations of associative algebras and dispersionless integrable                     
               hierarchies}, 
               Commun. Math. Phys. 274 (2007), 627-658 ; arXiv: nlin. SI/0606069 (2006) .

              
\bibitem{Witten} Witten E.,
              \emph{On the structure of topological phase of two-dimensional gravity},
               Nucl. Phys. B 340 (1990), 281-332.              

\bibitem{DVV} Dijkgraaf R.,Verlinde H. and Verlinde E., 
              \emph{Topological strings in d$<$1},
                Nucl. Phys. B 352 (1991), 59-86.

\bibitem{Dub} Dubrovin B.,
              \emph{Integrable systems in topological field theory},
                Nucl. Phys. B 3379 (1992), 627-689.

\bibitem{Dub2} Dubrovin B.,
              \emph{Geometry of 2D topological field theories},
               Lectures Notes in Math. 1620 (1996), 120-348 , Springer, Berlin.
               
\bibitem{Tsa} {Tsar\"ev, S. P.},
		\emph{The geometry of {H}amiltonian systems of hydrodynamic type.
		{T}he generalized hodograph method}, 
		{ Math. USSR-Izv. 37 (1991), no. 2, 397-419}.
 
\bibitem{FrNi} Fr\"olicher, A. and Nijenhuis, A.,
              \emph{Theory of vector-valued differential forms. {I}. {D}erivations
              of the graded ring of differential forms},
              {Nederl. Akad. Wetensch. Proc. Ser. A. {\bf 59} = Indag. Math.}, {18}, {(1956)}, {338--359}.
 
\bibitem{Haa} Haantjes, J.,
	     \emph{On {$X_m$}-forming sets of eigenvectors},
             {Nederl. Akad. Wetensch. Proc. Ser. A. {\bf 58} = Indag. Math.}, {17}, {(1955)}, {158--162}.
 
\bibitem{Ma} Magri F.,
	    \emph{WDVV equations},
	    {Il Nuovo Cimento},
	    {38 C}, (2015), 166\\
	    \texttt{DOI: 10.1393/ncc/i2015-15166-2}
	    
 
\bibitem{ON} O'Neill, Barrett,
	    \emph{Semi-{R}iemannian geometry},
	    {Pure and Applied Mathematics},
	    {103},
	    {Academic Press, Inc. [Harcourt Brace Jovanovich, Publishers], New York},
	    {1983}

\bibitem{Nij} Nijenhuis, Albert
             \emph{ Jacobi-type identities for bilinear differential concomitants of certain tensor fields. I}
             {Nederl. Akad. Wetensch. Proc. Ser. A. {\bf 58} = Indag. Math.}, {17}, {(1955)}, {390-397}.
 
\end{thebibliography}
\end{document}